\documentclass[superscriptaddress,aps,prx,nofootinbib,onecolumn,reprint,a4paper,longbibliography]{revtex4-2}

\usepackage{amssymb,amsthm,amsfonts,amstext,amsmath}
\usepackage[hidelinks]{hyperref}
\usepackage{bm}
\usepackage{color}

\newtheorem{theorem}{Theorem}

\newtheorem{lemma}{Lemma}
\newtheorem{Property}{Property}

\begin{document}
\title{Universal representation by Boltzmann machines with Regularised Axons}

\author{Przemys{\l}aw R. Grzybowski}
\thanks{Correspondence should be addressed to \href{mailto:grzyb@amu.edu.pl}{grzyb@amu.edu.pl}}
\affiliation{ICFO - Institut de Ciències Fot\`oniques, The Barcelona Institute of Science and Technology, 08860 Castelldefels (Barcelona), Spain}
\affiliation{Institute of Spintronics and Quantum Information, Faculty of Physics, Adam Mickiewicz University, Umultowska 85, 61-614 Pozna{\'n}, Poland}

\author{Antoni Jankiewicz}
\affiliation{Institute of Spintronics and Quantum Information, Faculty of Physics, Adam Mickiewicz University, Umultowska 85, 61-614 Pozna{\'n}, Poland}

\author{Eloy Pi\~nol}
\affiliation{ICFO - Institut de Ciències Fot\`oniques, The Barcelona Institute of Science and Technology, 08860 Castelldefels (Barcelona), Spain}

\author{David Cirauqui}
\affiliation{ICFO - Institut de Ciències Fot\`oniques, The Barcelona Institute of Science and Technology, 08860 Castelldefels (Barcelona), Spain}
\affiliation{Quside Technologies SL, Mediterranean Technology Park, 08860 Castelldefels (Barcelona), Spain}

\author{Dorota H. Grzybowska}
\noaffiliation{}

\author{Pawe{\l} M. Petrykowski}
\noaffiliation{}

\author{Miguel \'Angel Garc\'ia-March}
\affiliation{Instituto Universitario de Matem\'atica Pura y Aplicada, Universitat Polit\`ecnica de
Val\`encia, 46022 Val\`encia, Spain}

\author{Maciej Lewenstein}
\affiliation{ICFO - Institut de Ciències Fot\`oniques, The Barcelona Institute of Science and Technology, 08860 Castelldefels (Barcelona), Spain}
\affiliation{ICREA - Instituci\'o Catalana de Recerca i Estudis Avan\c cats, E-08010 Barcelona, Spain}

\author{Gorka Mu\~noz-Gil}
\affiliation{Institute for Theoretical Physics, University of Innsbruck, Technikerstr. 21a, A-6020 Innsbruck, Austria}

\author{Alejandro Pozas-Kerstjens}
\affiliation{Group of Applied Physics, University of Geneva, 1211 Geneva 4, Switzerland}
\affiliation{Constructor University, Geneva, Switzerland}
\affiliation{Instituto de Ciencias Matem\'aticas (CSIC-UAM-UC3M-UCM), 28049 Madrid, Spain}

\begin{abstract}
It is widely known that Boltzmann machines are capable of representing arbitrary probability distributions over the values of their visible neurons, given enough hidden ones.
However, sampling --and thus training-- these models can be numerically hard.
Recently we proposed a regularisation of the connections of Boltzmann machines, in order to control the energy landscape of the model, paving a way for efficient sampling and training.
Here we formally prove that such regularised Boltzmann machines preserve the ability to represent arbitrary distributions.
This is in conjunction with controlling the number of energy local minima, thus enabling easy \emph{guided} sampling and training.
Furthermore, we explicitly show that regularised Boltzmann machines can store exponentially many arbitrarily correlated visible patterns with perfect retrieval, and we connect them to the Dense Associative Memory networks.
\end{abstract}

\maketitle

\section{Introduction}\label{sec:intro}
Energy-based models are an attractive option for generative unsupervised \cite{NEURIPS2019_378a063b,Ramsauer2020HopfieldNI} and supervised learning \cite{kehoe2021defence}.
Among them, two closely related models, classical Hopfield networks~\cite{hopfield} and Boltzmann machines (BMs)~\cite{ACKLEY1985147}, have played a paradigmatic role in the development of the field.
The first serve as associative memory models, allowing for the storage and dynamical retrieval of system configurations, also known as \textit{patterns}.
However their storage capacity is proportional to the number of neurons $N$~\cite{Gardner_1988}, and therefore cannot represent arbitrary probability distributions over all the $2^N$ configurations of neurons.
The addition of latent (or \textit{hidden}) neurons led to BMs, which can be universal approximators of any probability distribution for the configurations of visible binary neurons~\cite{YOUNES1996109}.
Furthermore the shallow version of these models, called restricted BMs (RBM)~\cite{smolensky1986information,freund1991unsupervised}, was proven to be easily trainable, and thus became standard elements of many larger networks such as deep belief networks~\cite{bengio2009learning}.
Unfortunately, training deep BMs \cite{pmlr-v5-salakhutdinov09a} requires substantial numerical effort because sampling these models is hard.
Although some advanced Markov Chain Monte Carlo (MCMC) methods have been adopted for that end, such as Coupled Adaptive Simulated Tempering~\cite{10.5555/3104322.3104442}, deep BMs are less widely used than RBMs despite their potential advantage in learning data features.

Recently, Ref.~\cite{rapid} proposed to regularise the weights (axons) of BMs in order to control the energy landscape of the model, and ultimately to ease its sampling and training.
This is achieved by building the model's weights from combinations of a set of patterns, similar to the Hebbian weights in Hopfield networks.
Then, it is these patterns, instead of the weights, who constitute the trainable parameters of the model.
The inspiration for this regularisation stems from controlling the ``spin frustration'' of the model, as well as making the weight matrix similar to that of successfully trained BMs~\cite{Decelle_2017,doi:10.7566/JPSJ.91.114001}.
The successful regularisation of weights should result in a controlled number of energy local minima, such that the patterns are (or are close to) low-energy states of the model.
Therefore, the information about patterns should allow for an effective \emph{guided} sampling of the model at low temperatures.
This hypothesis was empirically verified in Ref.~\cite{rapid} for several benchmark data sets.
Yet, the values that weights can take in a regularised BM (RA-BM, using the terminology of \cite{rapid}) is a strict subset of the values that weights of unregularized BMs can take.
This motivates the question of the real representation power of regularised BMs.

In this work, we present a formal proof that regularised BMs can represent arbitrary probability distributions over the values of their visible binary neurons.
 In addition to this, and importantly, we formally prove the intuitions verified in \cite{rapid} regarding sampling: RA-BMs can always have easy pattern-guided (or PID, in the terminology of \cite{rapid}) sampling.
 These results allow us to prove that the construction of weights from patterns in BMs can work very differently from the Hebbian weights of the Hopfield networks.
In the standard Hopfield model (cf. \cite{Amit-book}), patterns are uncorrelated and are recovered faithfully only if their number is smaller than $0.14 N$~\cite{Amit1,Amit2}. Furthermore, in Hopfield networks, additional minima corresponding to ``spurious states'' arise from mixed patterns.

A simple ``signal-to-noise'' estimation \cite{Amit-book}, or an estimate based on Random Matrix Theory \cite{Tarkowski2} shows that storage capacity in the Hopfield model with Hebbian rule is also non-zero and proportional to $N$. In fact, various kinds of correlated patterns have been considered in the literature: ``spatially'' \cite{Tarkowski1,Tarkowski2,Tarkowski4,Monasson1,Monasson2}, ``semantically'' \cite{Tarkowski2,Tarkowski4}, and ``symmetry invariant'' \cite{Tarkowski3} correlated patterns. 
These papers, with the exception of \cite{Tarkowski2}, mostly focused on optimal storage capacity using Gardner's approach \cite{Gardner_1988}, but the techniques used clearly indicate that similar calculation is feasible for the problem of storage capacity of the Hopfield model with the Hebbian rule with correlated patterns, at least for weak correlations \cite{Grzybowski-future}. We stress, however, that here, we show that for BMs at low temperatures, patterns can be perfectly stored even if they are arbitrarily strongly correlated in their visible part. For structured data, one can consider two-layered network as in the Bidirectional Associative Memories (BAM), introduced by Kurchan {\it et al.} in the 1990s \cite{Kurchan1994} which exhibit increased storage capacity than analogous Hopfield network with Hebbian rule, in case of data asymmetry. These results has been generalized recently to finite temperature \cite{Barra2023}. Seemingly the mathematical form of BAM networks is similar to RA-RBMs. Yet storage capacity of uncorrelated patterns for BAM and Hopfield models alike is bound by $2N$ \cite{Gardner_1988} ({\it i.e.} is linear). The use of latent (hidden) units in BMs allows to overcome this limitation. Here we prove that RA-RBMs can store the number of patterns as large as $2^V$, where $V$ is the number of visible neurons. We also prove the absence of spurious states at low temperatures irrespective of patterns correlations.
 Interestingly, these properties resemble Dense Associative Memory networks or modern Hopfield networks~\cite{Ramsauer2020HopfieldNI,krotov,krotov2021large,mezard}.
 Indeed, should one marginalize a BM over its hidden neurons, the result would be a network consisting only of $V$ visible neurons with multi-spin interactions, in the spirit of Refs.~\cite{krotov, Demircigil}.
Should one marginalize an RA-BM over hidden neurons, the resulting network, at low temperatures, can have the exponential capacity for storage and perfect retrieval of arbitrarily correlated patterns. In this sense, BMs with regularized weights are related to the embeddings of Dense Associative Memory networks in a set of two-body interaction models~\cite{krotov2021large}.  

The work is organized as follows: in Section \ref{sec:prevwork} we discuss related literature; in Section \ref{sec:prelim} we set the notation and definitions for the rest of the manuscript; Section \ref{sec:proofs} contains the main body of the work, presenting the model construction and proving the relevant theorems, which are discussed in detail in Section \ref{sec:discussion}; finally, Section \ref{sec:summary} contains a summary of the results and an outlook of their impact.

\section{Related work}\label{sec:prevwork}
In general, the representational power of BMs depends on its construction, i.e. the number of hidden neurons and the topology of their connections, as well as on the correlations in the data set that should be learnt.
For restricted BMs it is known that: (i) each hidden unit can model the probability of one visible configuration (called also ``elementary event'')~\cite{le2008representational}; (ii) each hidden unit can model the probability of two visible configurations of Hamming distance one~\cite{montufar2011refinements}; (iii) each hidden unit can model a block of correlated visible configurations with weighted probability distribution~\cite{montufar2011expressive}.
The representational power of restricted BMs was also considered from the perspective of its representation efficiency~\cite{martens2013representational}.
In the most challenging scenario, i.e. when visible configurations are assigned random probabilities, the information about the dataset's probability distribution cannot be compressed and the BM needs, in order to properly represent the data, about $2^V$ (where $V$ is the number of visible binary neurons) parameters.
This entails the need of a number of hidden neurons exponential in $V$, $H\sim 2^V$, if one aims to represent \textit{any} probability distribution.
Additionally, in such worst-case scenario, deep BMs do not show any representational advantage over restricted machines, as there are no features which can be abstracted nor any underlying structure can be learned.

The RAPID approach developed in Ref.~\cite{rapid}, primarily consisting in a regularisation of weights (RA), aims at controlling the energy landscape of the model and with it the associated sampling difficulty.
Since the phase space is exponentially large with respect to the total number of neurons $N$, sampling is usually performed by some version of MCMC simulation.
In the case of random independent weights (which is a standard way to initialise the BM models \cite{HintonRBMguidelines}), BMs and classical Hopfield models are equivalent to the non-planar Ising spin-glass model~\cite{RevModPhys.58.801}.
Such systems, at low temperatures, exist in a spin-glass phase where sampling is an NP-complete problem~\cite{Barahona1982complexity}.
This complexity can be traced back to the existence of an exponential number of energy local minima (with respect to $N$), with comparable low energies and separated by high-energy barriers.

The regularisation in RAPID controls the number of low-energy minima by building the weights form patterns, such that: i) patterns should be proxies for configurations corresponding to low-energy minima; and ii) the total number of minima should be proportional to the number of patterns.
Yet, the mathematical form of the Hebbian weights alone does not suffice to guarantee any of these claims.
This is evident in case of classical Hopfield networks, where the system enters a spin-glass phase with exponentially many minima if the number of patterns is greater that $0.14 N$~\cite{Amit1,Amit2}.
Therefore, first we present Theorem \ref{thm:perfect} stating that the RAPID regularisation of BMs can perfectly memorise an arbitrary number of patterns with arbitrarily correlated visible parts.
To this end we present an explicit mathematical construction for a weight-regularised restricted BM whose only energy minima are the aforementioned patterns (in other words, allowing for perfect pattern retrieval).
Importantly, these properties also assure that pattern-guided sampling relevant for regularised BMs can be fast, efficient and precise.
We prove in Theorem \ref{thm:representing} that in such construction one can assign arbitrary probabilities to configurations associated with patterns, using only as many hidden variables as patterns.
This statement draws similarities to those used in the proofs of universal representability of RBMs \cite{le2008representational}, which showed that each hidden unit can model the probability of one visible configuration.

\section{Preliminaries and notation}\label{sec:prelim}
Let us begin by recalling standard restricted BMs, consisting of a bipartite structure with $V$ visible neurons--of which a given configuration is denoted by $\bm{v}=\{v_j\}_{j=1}^V\in\{-1,1\}^V$--and $H$ hidden neurons--for which a configuration is denoted by $\bm{h}=\{h_\alpha\}_{\alpha=1}^H\in\{-1,1\}^H$--, with every neuron in each layer connected to all the neurons in the other, and no neuron connected to any other in the same layer.
From now onward we use the convention that roman subindices always enumerate from $1$ to $V$ while greek subindices always enumerate from $1$ to $H$.
Each configuration is assigned an energy via
\begin{equation}
    E(\bm{v},\bm{h})=-\sum_{j=1}^V\sum_{\alpha=1}^H v_j W_{j,\alpha} h_\alpha -\sum_{j=1}^V b_j v_j -\sum_{\alpha=1}^H c_\alpha h_\alpha,
    \label{eq:energy}
\end{equation}
where $\theta=(\bm{W},\bm{b},\bm{c})$ are the parameters of the model.
The weight matrix $\bm{W}$ describes couplings between neurons and the vectors $\bm{b}$ and $\bm{c}$ describe biases (or, in the context of statistical physics, external fields) acting on, respectively, visible and hidden neurons.
The energy function enables to assign probabilities to system configurations according to the Boltzmann distribution:
\begin{equation}
    P_\theta(\bm{v},\bm{h})=\frac{e^{-E_\theta(\bm{v},\bm{h})}}{\sum_{\bm{v},\bm{h}}e^{-E_\theta(\bm{v},\bm{h})}},
    \label{eq:Boltzmann}
\end{equation}
which upon marginalisation over hidden variables gives the probability of the model producing a given visible configuration:
\begin{equation}
    P_\theta(\bm{v})=\frac{\sum_{\bm{h}}e^{-E_\theta(\bm{v},\bm{h})}}{\sum_{\bm{v},\bm{h}}e^{-E_\theta(\bm{v},\bm{h})}}.
    \label{eq:BoltzmannV}
\end{equation}
The general task of learning with restricted BMs is defined as finding appropriate parameters $\theta$ such that $P_\theta$ captures the probability distribution underlying some dataset.

For a more detailed study of the representative power of regularised RBMs, in our construction we define as $\Omega$ the subset of visible configurations for which the probability distribution will be modeled, i.e. the training dataset.
We refer to $\bm{v}\in\Omega$ as \textit{represented} configurations and, for reasons that will become evident later on, we will choose the number of hidden neurons to satisfy $H=|\Omega|$.
Also it will prove convenient to label represented configurations with greek superindices as $\bm{v}^\alpha$, i.e. $\Omega=\{\bm{v}^\alpha\}_{\alpha=1}^H$.
For every visible configuration $\bm{v}$ there is exactly one configuration which is its inversion, $-\bm{v}$.
We will use the notation $\alpha$ and $-\alpha$ to refer to a pair of mutually inverted configurations in $\Omega$, $\bm{v}^\alpha$ and $-\bm{v}^\alpha$.
For the dot product between two visible represented configurations indexed $\alpha$ and $\beta$ we use the shorthand:
\begin{align}
    \label{eq:Product}
    S_{\alpha,\beta}=\bm{v}^\alpha\cdotp\bm{v}^\beta=\sum_{j} v^\alpha_j v^\beta_j.
\end{align}
We will use the fact that values of $S_{\alpha,\beta}$ form a set of integers between $-V$ and $V$ (included) with increment of 2.
Finally, by $\bm{v}'$ we will denote an arbitrary visible configuration that is not represented i.e. $\bm{v}'\not\in\Omega$.

Since the number of local minima of the energy function will play an important role in our considerations, let us introduce a bound:
\begin{lemma}
    \label{thm:RBMbound}
    For any restricted BM the number of local minima of the energy function is bounded from above by $\min(2^V, 2^H)$.
\end{lemma}
\begin{proof}
Consider system configurations of the form $(\bm{x},\bm{h})$ where the visible configuration is fixed to some $\bm{v}=\bm{x}$ while $\bm{h}$ is arbitrary.
The energy of configuration $(\bm{x},\bm{h})$ reads:
\begin{align}\label{eq:model}
    E(\bm{x},\bm{h}) &= -\sum_{\alpha=1}^H (\sum_{j=1}^V x_j W_{j,\alpha} + c_\alpha)h_\alpha +C_0,
\end{align}
where $C_0=-\sum_{j=1}^V b_j x_j $ does not depend on $\bm{h}$.
The energy depends on $\bm{h}$ through a sum over the values of the hidden neurons multiplied by effective fields of strength $(\sum_{j=1}^V x_j W_{j,\alpha} + c_\alpha)$.
The lowest-energy configuration of the form $(\bm{x},\bm{h})$ is thus attained by having the hidden neurons align with their respective effective fields.
The corresponding configuration belongs to a \emph{possible} local minimum of the energy: if some effective fields are zero, there are several configurations with the same energy, all pertaining to one possible local minimum.
Whenever a given possible local minimum is indeed a minimum, there can be no other local energy minima among system configurations of the form $(\bm{x},\bm{h})$.
This is because any configuration of the form $(\bm{x},\bm{h})$ with hidden spins not aligned to their respective effective local fields can have energy lowered by single spin-flips of those hidden spins.
In other words, there is at most one local minimum per visible configuration.
Since there are $2^V$ visible configurations and any configuration has to have a visible part, we get a bound of $2^V$ on the number of local minima of the energy function.
Due to the bipartite structure of RBMs, the same analysis can be performed for system configurations when the hidden neurons are fixed and the visible neurons are free, obtaining in that case the bound of $2^H$.
\end{proof}

As commented in Section \ref{sec:intro}, Ref.~\cite{rapid} introduced a method to control the number of energy minima of restricted BMs by regularizing the weight matrix $\bm{W}$.
This was achieved by constructing $\bm{W}$ from a set of system configurations, called patterns, $\{\bm{\xi}^\eta\in\{-1,1\}^{V+H}\}_{\eta=1}^K$:
\begin{equation}
    W_{j,\alpha} = A\sum_{\eta=1}^K \xi^\eta_j\xi^\eta_{\alpha},
    \label{eq:ra}
\end{equation}
where $\xi^\eta_j$ and  $\xi^\eta_{\alpha}$ are the values of visible and hidden neurons in pattern $\bm{\xi}^\eta$, pertaining to sites $j$ and $\alpha$ respectively. The total number of patterns is denoted as $K$. The amplitude $A$ adjusts the typical size of the weights, therefore it also controls the ``temperature''.
Values $A>1$ correspond to low effective temperatures, while values $A<1$ denote high effective temperatures.
While the visible parts of the patterns can be related to the data, the hidden parts are not (in contrast, in Hopfield models all patterns are related to the data being modeled), and their only role is to structure the energy landscape.
In practice, when training restricted BMs where Eq.~\eqref{eq:ra} applies the trainable parameters are not the weights $\bm{W}$ but the patterns $\bm{\xi}^\eta$.
However, in the construction that we present here we will match the visible parts of the patterns to the represented visible configurations, and then choose the hidden parts accordingly.

\section{Construction and proofs}\label{sec:proofs}
This section contains the core of the work, namely (i) the proofs that the restricted machines proposed in Ref.~\cite{rapid} have the ability of universal approximation, and (ii) its connection to Dense Associative Memories.

In order to begin, we impose a reflection symmetry constraint on the set of represented visible configurations $\Omega$: $\bm{v}\in\Omega\Leftrightarrow -\bm{v}\in\Omega$.
This is, we consider the case where visible configurations are represented in pairs, labeled $\alpha, -\alpha$.
Although this may initially seem to be a restriction in the families of models that we consider, the probabilities that are assigned, $P_\theta(\bm{v}^\alpha)$ and $ P_\theta(-\bm{v}^\alpha)$, can be different and arbitrary.
Therefore, this choice merely simplifies calculations at the cost of a possible (linear) overhead in terms of hidden neurons.
This overhead may arise if, for some pair of configurations, $P_\textrm{data}(\bm{v}^\alpha) \gg P_\textrm{data}(-\bm{v}^\alpha)$ because data configurations with very low probability are usually not modeled.
Still, our construction requires the additional hidden unit for $-\bm{v}^\alpha$.
Note, nonetheless, that the overhead vanishes when we consider the worst-case scenario for data, where all configurations in $\{-1,+1\}^V$ have non-zero probability.

We now proceed to give a construction that allows to assign, to any dataset $\Omega=\{\bm{v}^\eta\}_{\eta=1}^H$ such that $\bm{v}\in\Omega\Leftrightarrow -\bm{v}\in\Omega$ and with associated probabilities given by $P_\textrm{data}(\bm{v}^\eta)$, an expression of such distribution in the form of Eq.~\eqref{eq:BoltzmannV} with the weight matrix satisfying Eq.~\eqref{eq:ra}.
Without loss of generality, we set amplitude $A>0$, visible biases to zero, ${\bm b} = 0$, and choose the number of hidden neurons and patterns $K=H=|\Omega|$.
We also choose the patterns to have the values:
\begin{align}
    \label{eq:FCpatterns}
    \xi^\eta_j=v^\eta_j,\ \ \ \xi^\eta_\alpha=2\delta_{\alpha \eta}-1,
\end{align}
where $\delta_{a b}$ is the Kronecker delta.
The visible part of pattern $\bm{\xi}^\eta$ is set to the \emph{represented} data configuration $\bm{v}^\eta$ whose probability we aim to model.
The hidden part of the pattern $\bm{\xi}^\eta$, to which we will refer as $\bm{h}^\eta$, has uniform values $-1$ except for the site $\eta$, making the hidden parts of the patterns strongly correlated.
We choose this form of hidden parts of patterns for convenience in the proofs, but they can be chosen in many ways, depending on the particular construction.

Using this condition and substituting Eq.~\eqref{eq:FCpatterns} into weight $W_{j,\alpha}$ given by Eq.~\eqref{eq:ra} we find that, for every pair of represented configurations, $\beta$ and $-\beta$, that is not the pair $\{\alpha,-\alpha\}$, the corresponding contributions to $W_{j,\alpha}$ have opposite signs on visible site $j$ (i.e. $\xi^\beta_j=-\xi^{-\beta}_j$) but the same sign on hidden site $\alpha$ (this is, $\xi^\beta_\alpha=\xi^{-\beta}_\alpha=-1$).
In result, the sum of the contributions to $W_{j,\alpha}$ from such pair is zero.
Only the pair $\{\alpha,-\alpha\}$ has a non-zero contribution to $W_{j,\alpha}$, resulting in
\begin{align}
    \label{eq:FCweightsfinal}
    W_{j,\alpha}=A(v^\alpha_j-v^{-\alpha}_j)=2A v^\alpha_j.
\end{align}
This is, all weights take values $\pm 2A$.
Next, we re-scale hidden biases by the same factor $2A$ and express them as:
\begin{align}
    \label{eq:FCcfinal}
    c_{\alpha}=2A (\lambda_\alpha - V),
\end{align}
imposing the additional constraint $\frac34<\lambda_\alpha <\frac54$.
The consequences and the reasons for this constraint will become clear later on.

In the following we prove five important properties related to the construction we propose.
Properties \ref{thm:FCunique1}, \ref{thm:FCunique10} and \ref{thm:FCunique4} assert that its regularisation is effective in controlling the local minima structure of the model.
Additionally, as we formally state in Theorem \ref{thm:perfect}, this can be seen as a proof of the storage power of general regularised BMs.
Then, properties \ref{thm:FCunique3} and \ref{thm:FCunique2} allow us to give a proof of universal representability of general regularised RBMs, formulated in Theorem \ref{thm:representing}.
In this section, we focus mostly on the technical aspects of these properties and theorems, while general discussion is deferred to Section~\ref{sec:discussion}.
\begin{Property}
    \label{thm:FCunique1}
    The patterns given by Eq.~\eqref{eq:FCpatterns} of the regularised RBM model given by Eqs.~\eqref{eq:model}-\eqref{eq:ra} are local minima of the associated energy function.
\end{Property}
\begin{proof}
    The proof will rely on showing that any single spin flip from any given pattern, $\bm{\xi}^\eta = (\bm{v}^\eta,\bm{h}^\eta)$, increases its energy.
    First consider the energy of any configuration $(\bm{v}^\eta,\bm{h})$ where $\bm{h}$ is arbitrary.
    Using Eqs.~\eqref{eq:FCweightsfinal} and \eqref{eq:FCcfinal} we get
    \begin{align}
        \label{eq:FCEcondV}
        E(\bm{v}^\eta,\bm{h}) &= - \sum_{j,\alpha}W_{j,\alpha}v_j^\eta h_\alpha -\sum_\alpha c_\alpha h_\alpha \nonumber \\ &=- \sum_{j,\alpha}2 A v^\alpha_j v^\eta_j h_{\alpha}-\sum_\alpha c_\alpha h_\alpha \nonumber \\ 
        &= - \sum_{\alpha} 2A\Bigr(S_{\alpha,\eta}-V +\lambda_\alpha    \Bigl)h_\alpha.
    \end{align}
    Note that this is a sum of the values of the hidden neurons multiplied by local effective fields given by $2A(S_{\alpha,\eta}-V+\lambda_\alpha)$.
    For $\alpha\neq\eta$ we have $S_{\alpha,\eta}-V\leq-2$ and since $\frac34<\lambda_\alpha <\frac54$, the local field is negative.
    In contrast, for $\alpha=\eta$ the local field reduces to $\lambda_\eta$, which is positive.
    Therefore, there is only one hidden configuration,
    $(\bm{h}|\bm{v}^\eta)_{\text{min}}$,
    that minimizes $E(\bm{v}^\eta,\bm{h})$ by aligning to local fields, namely
    \begin{align}
        \label{eq:FCHminV}
        (h_{\alpha}|\bm{v}^\eta)_{\text{min}}=2\delta_{\alpha \beta}-1\equiv\xi^\eta_\alpha,
    \end{align}
    that is, the hidden neuron configuration contained in the pattern $\eta$.
    Furthermore, flipping the hidden spin $h_\alpha$ in the pattern $\eta$ results in a configuration with energy increased by twice the magnitude of the local effective field acting at site $\alpha$:
    \begin{align}
        \label{eq:FCdeltah}
        \Delta E=4A\mid S_{\alpha,\eta}-V+\lambda_\alpha\mid >3A.
    \end{align}
    This inequality follows from $\frac34<\lambda_\alpha <\frac54$ and holds regardless of $\alpha\neq\eta$ or $\alpha=\eta$.

    Consider now the energy of any configuration $(\bm{v},\bm{h}^\eta)$ where $\bm{v}$ is arbitrary.
    Using Eqs.~\eqref{eq:FCpatterns}, \eqref{eq:FCweightsfinal} and \eqref{eq:FCcfinal} we find:
    \begin{align}
        \label{eq:FCEcondH}
        E(\bm{v},\bm{h}^\eta) = - \sum_{j,\alpha}W_{j,\alpha}v_j (2\delta_{\alpha \eta}-1) -\sum_\alpha c_\alpha(2\delta_{\alpha \eta}-1) \nonumber \\ 
        =- \sum_{j}4A v^\eta_j v_j+ \sum_{j,\alpha} 2A v^\alpha_j v_j-\sum_\alpha c_\alpha(2\delta_{\alpha \eta}-1).
    \end{align}
    The second term is zero since every pair of mutually inverted, represented, visible configurations, $\alpha$ and $-\alpha$, have opposite signs on the visible site $j$.
    Note also that the third term does not depend on the configuration $\bm{v}$.
    The first term is a sum over the values of the visible neurons, $v_j$, multiplied by local effective fields given by $4Av^\eta_j$.
    Therefore, there is only one visible configuration
    $(\bm{v}|\bm{h}^\eta)_{\text{min}}$ that minimizes $E(\bm{v},\bm{h}^\eta)$, namely the one obtained by aligning to the local fields:
    \begin{align}
        \label{eq:FCVminH}
        (v_{j}|\bm{h}^\eta)_{\text{min}}=v_j^\eta.
    \end{align}
    This is the configuration of the visible neurons contained in the pattern $\eta$.
    Flipping the visible spin $v_j$ in pattern $\eta$ results in configuration with energy increased by twice the magnitude of the local effective field acting at site $j$, namely
    \begin{align}
        \label{eq:FCdeltaj}
        \Delta E=8A\mid v_j^{\eta}\mid =8A.
    \end{align}
\end{proof}
\begin{Property}
    \label{thm:FCunique10}
    For a given visible, not-represented configuration $\bm{v}'\notin\Omega$, the lowest energy attainable by configurations of the form $(\bm{v}',\bm{h})$, where $\bm{h}$ is arbitrary, is equal to $E_0=-2A\sum_{\alpha} (V-\lambda_\alpha)$.
\end{Property}
\begin{proof}
    Using Eqs.~\eqref{eq:FCweightsfinal} and \eqref{eq:FCcfinal} we get:
    \begin{align}
        \label{eq:FCEcondVprim}
        E(\bm{v}',\bm{h}) &= - \sum_{j,\alpha}W_{j,\alpha}v'_j h_\alpha -\sum_\alpha c_\alpha h_\alpha \nonumber \\ &=- \sum_{j,\alpha}2 A v^\alpha_j v'_j h_{\alpha}-\sum_\alpha c_\alpha h_\alpha \nonumber \\ 
        &= - \sum_{\alpha} 2A\Bigr(\bm{v}'\cdot\bm{v}^\alpha-V +\lambda_\alpha    \Bigl)h_\alpha.
    \end{align}
    Again, we find a sum over values of hidden neurons multiplied by local effective fields given by $2A(\bm{v}'\cdot\bm{v}^\alpha-V+\lambda_\alpha)$.
    Since $\bm{v}^\alpha\neq\bm{v}'$ for any $\alpha$ we have $\bm{v}'\cdot\bm{v}^\alpha-V\leq-2$.
    In combination with $\frac34<\lambda_\alpha <\frac54$, we see that all local fields are negative.
    Therefore, there is only one hidden configuration $(\bm{h}|\bm{v}')_{\text{min}}$ that minimizes $E(\bm{v}',\bm{h})$ by aligning to local fields, which has elements $ (h_{\alpha}|\bm{v}')_{\text{min}}=-1$ for all $\alpha$.
    We denote such hidden configuration as $-\bm{e}$.
    
    Let us now consider the energy of configuration $(\bm{v},-\bm{e})$ where $\bm{v}$ is arbitrary.
    Using Eqs.~\eqref{eq:FCweightsfinal} and \eqref{eq:FCcfinal} we get:
    \begin{align}
        \label{eq:FCEcondHprim}
        E(\bm{v},-\bm{e})&= - \sum_{j,\alpha}W_{j,\alpha}v_j (-1) -\sum_\alpha c_\alpha(-1) \nonumber \\ 
        &=\sum_{j,\alpha} 2A v^\alpha_j v_j+\sum_{\alpha} 2A\Bigr(\lambda_\alpha-V    \Bigl)\nonumber \\ 
        &=2A\sum_{\alpha}\bm{v}^\alpha \cdot \bm{v} +E_0 = E_0,
    \end{align}
    where $E_0=-2A\sum_{\alpha} (V-\lambda_\alpha)$.
    With the condition $\frac34<\lambda_\alpha <\frac54$ we see that $E_0<0$ if $V\geq 2$.
    The first term in the last line of Eq.~\eqref{eq:FCEcondHprim} is zero due to cancellations between every pair of mutually inverted, represented, visible configurations.
    Note that the energy of configuration $(\bm{v},-\bm{e})$ does not depend on $\bm{v}$.
    This implies that $E_0$ is the smallest energy attainable by the configurations of the form $(\bm{v}',\bm{h})$.
\end{proof}
\begin{Property}
    \label{thm:FCunique4}
    There are no other local minima of the energy function in the regularised RBM model of Eqs.~\eqref{eq:model}-\eqref{eq:ra} than the patterns given by Eq.~\eqref{eq:FCpatterns}.
\end{Property}
\begin{proof}
    For this proof we will use properties \ref{thm:FCunique1} and \ref{thm:FCunique10} to consider two complementary cases: the case of the form $(\bm{v}^\eta,\bm{h})$ for any $\eta$ and where $\bm{h}$ is arbitrary, and the case of the form $(\bm{v}',\bm{h})$ for any $\bm{v}'$ and where $\bm{h}$ is again arbitrary.
    
    In the first case, using Eqs.~\eqref{eq:FCEcondV}-\eqref{eq:FCHminV}, we conclude that pattern $\bm{\xi}^\eta$ is the only configuration whose energy is a local minimum among all configurations of the form $(\bm{v}^\eta,\bm{h})$.
    
    In the second case, from Eq.~\eqref{eq:FCEcondVprim} we conclude that among all the configurations of the form $(\bm{v}',\bm{h})$, only the configuration $(\bm{v}',-\bm{e})$ can be a candidate for a local minimum of the energy.
    However, this configuration is not a local minimum because we can lower the energy of the system by a sequence of single spin flips, starting from configuration $(\bm{v}',-\bm{e})$, without ever increasing energy during that process.
    As shown by Property \ref{thm:FCunique10} and Eq.~\eqref{eq:FCEcondHprim}, the energy of the configuration $(\bm{v}',-\bm{e})$ is equal to $E_0$, and so is the energy of any other configuration of the form $(\bm{v},-\bm{e})$.
    Therefore, using single visible-spin flips we can transition, with no energy cost, to configuration $(\bm{v}^\eta,-\bm{e})$ for some $\eta$.
    The latter differs from pattern $\bm{\xi}^\eta=(\bm{v}^\eta,\bm{h}^{\eta})$ only in the hidden spin at position $\eta$.
    Since pattern $\bm{\xi}^\eta$ is a local minimum by Property \ref{thm:FCunique1}, any configuration differing by a single spin flip must have higher energy.
    Thus, flipping the hidden spin at position $\eta$ in configuration $(\bm{v}^\eta,-\bm{e})$ must lower its energy.
\end{proof}

From the properties above, we introduce now our first main result:
\begin{theorem}
    \label{thm:perfect}
    A Boltzmann machine with regularised weights as in Eq.~\eqref{eq:ra} can perfectly memorise any set of visible configurations $\Omega$ that satisfies $\bm{v}\in\Omega\Leftrightarrow -\bm{v}\in\Omega$.
    For that it suffices to use $|\Omega|$ hidden neurons, i.e. one hidden unit per memorised visible configuration.
\end{theorem}
\begin{proof}
Here perfect memorisation of the set is understood as the storage and the perfect retrieval of any of it elements.
Since $\Omega$ contains only visible configurations, we use these as the visible parts of system configurations.
More concretely, the configurations we choose are the set of all patterns given by Eq.~\eqref{eq:FCpatterns} whose visible parts match the elements of $\Omega$.
Furthermore we set all $\lambda_\alpha=\frac{6}{5}$.

Without loss of generality it is sufficient to consider regularised RBMs since every general BM can be reduced to an RBM.
The construction proposed here enables perfect memorisation of the set $\Omega$.
This is, under zero-temperature Markov Chain Monte Carlo (MCMC) dynamics: (i) if the visible part of a configuration is stored in the model, it will not change, (ii) if the visible part of a configuration is not stored in the model, it will change randomly under MCMC dynamics until a stored configuration is reached, (iii) for a configuration $\eta$ whose visible part is stored except for a single \emph{undetermined} spin (i.e., $v_j=0$ for one single $j$), the dynamics can be used to retrieve the stored visible configuration.
For this point we explicitly assume that there is no other stored configuration $\eta'$ that differs from $\eta$ {\it only} at site $j$, since this information already determines the value of $v_j$ for $\eta$.

For showing (i), we note that if the current configuration is equal to one of the patterns, Property \ref{thm:FCunique1} ensures that its energy is a local minimum of Eq.~\eqref{eq:model}, and hence it is also a fixed point of the zero-temperature MCMC dynamics.
Further, consider starting in a configuration where the visible part matches a memorised configuration $\eta$, $(\bm{v}^\eta,\bm{h})$, but the hidden part is arbitrary.
Using Gibbs sampling, after the first sampling of hidden variables these align with the local fields (recall Eq.~\eqref{eq:FCEcondV}) and the hidden configuration changes into the hidden part of pattern $\eta$ (Eq.~\eqref{eq:FCHminV}).
Thus, the configuration is now pattern $\eta$, and we reached a fixed point of the dynamics.
Note that only the hidden part of the original configuration has changed.

For demonstrating (ii), consider any configuration whose visible part is not a memorised state, $(\bm{v}',\bm{h})$, and the hidden part is arbitrary.
After sampling the hidden neurons for the first time, they align with local fields (recall Eq.~\eqref{eq:FCEcondVprim}), so the hidden configuration becomes $-\bm{e}$.
Since the energy of any system configuration of the form $(\bm{v},-\bm{e})$ with arbitrary $\bm{v}$ is always $E_0$ per Property \ref{thm:FCunique10}, there is no energy gradient for the subsequent sampling of visible neurons, leading to a random visible configuration.
If such random visible configuration happens to match a stored visible configuration we arrive at point (i).
Otherwise, we obtain a new state of the form $(\bm{v}',\bm{h})$, where now $\bm{h}=-\bm{e}$ but $\bm{v}'$ is still not memorised.
Repeating the Gibbs sampling procedure eventually leads to a stored configuration.

For proving (iii), consider a visible input $\bm{v}_{\textrm{in}}$ that is equal to $\bm{v}^\eta$ for some $\eta$ except for a single undetermined visible spin at position $j$.
This is, the value of this spin is $0$.
Note that this is not a valid configuration (recall that $v_j\in\{-1,1\}$), but can be set as starting point for MCMC dynamics.
Consider now the energy of the input configuration $(\bm{v}_{\textrm{in}},\bm{h})$ whose hidden part is arbitrary: 
\begin{align}
    \label{eq:FCEcondVany}
    E(\bm{v}_{\textrm{in}},\bm{h}) = - \sum_{\alpha} 2A\Bigr(\bm{v}_{\textrm{in}}\cdot\bm{v}^\alpha-V +\lambda_\alpha    \Bigl)h_\alpha.
\end{align}
This has the form of a sum over the value   s of hidden neurons multiplied by local effective fields given by $\bm{v}_{\textrm{in}}\cdot\bm{v}^\alpha-V +\lambda_\alpha$.
The value of $\bm{v}_{\textrm{in}}\cdot\bm{v}^\alpha - V$ is equal to $-1$ for $\alpha=\eta$ (the $j$-th site contributes $0$ to the scalar product) and is no greater than $-3$ for any $\alpha\neq\eta$.
After the first sampling of hidden neurons, these align with local fields that are negative for all $\alpha\neq\eta$ except for $\alpha=\eta$.
In this case, as we constrained $\lambda_\alpha=\frac{6}{5}$, the local fields are positive.
Therefore, the configuration of the hidden neurons changes into $\bm{h}^\eta$, the hidden part of pattern $\eta$.
Then, the subsequent sampling of the visible neurons will restore the visible configuration $\bm{v}^\eta$ according to Eq.~\eqref{eq:FCEcondH}, and the whole system will have reached a fixed point of the dynamics.
\end{proof}

Next, we introduce two new properties which will help us to show that one can assign arbitrary probabilities to the memorised patterns.
We introduce also the symbol $\mathcal{E}$ to denote the energy of a pattern, i.e., $\mathcal{E}_\eta\equiv E(\bm{v}^\eta,(\bm{h}|\bm{v}^\eta)_{\text{min}})$.
\begin{Property}
    \label{thm:FCunique3}
    The difference in energy between patterns, and their probability ratios, depends only on differences of parameters $\lambda$ multiplied by $4A$ and read:
\begin{align}
    \label{eq:FCpatternrelative}
    \mathcal{E}_\alpha - \mathcal{E}_\beta = 4A(\lambda_\beta -\lambda_\alpha),\nonumber\\ 
    \frac{P(\bm{\xi}^\alpha)}{P(\bm{\xi}^\beta)}=\frac{e^{ -\mathcal{E}_\alpha}}{ e^{-\mathcal{E}_\beta}}=e^{-4A(\lambda_\beta -\lambda_\alpha)}.
\end{align}
\end{Property}
\begin{proof}
Let us explicitly write the energy of a pattern $\eta$ using last line of Eq.~\eqref{eq:FCEcondV}:
\begin{align}
    \label{eq:FCpatternEn}
    \mathcal{E}_\eta &\equiv 
    E(\bm{v}^\eta,(\bm{h}|\bm{v}^\eta)_{\text{min}}) \nonumber \\
    &= - \sum_{\alpha} 2A\Bigr(S_{\alpha,\eta}-V +\lambda_\alpha    \Bigl)(2\delta_{\alpha \eta}-1)\nonumber\\
    &=-4A\lambda_\eta + \sum_{\alpha} 2A \Bigr(S_{\alpha,\eta}-V +\lambda_\alpha    \Bigl)\nonumber\\
    &=-4A\lambda_\eta+E_0,
\end{align}
where we used the fact that $\sum_{\alpha} S_{\alpha,\eta}=0$ due to cancellations occurring in pairs of mutually inverted, represented, visible configurations.
Then, Equation \eqref{eq:FCpatternrelative} follows directly from Eq.~\eqref{eq:FCpatternEn} and Eq.~\eqref{eq:Boltzmann}.
\end{proof}
\begin{Property}
    \label{thm:FCunique2}
   The energy of any pattern is lower than the energy of any non-pattern configuration by a magnitude of, at least, $A$.
\end{Property}
\begin{proof}
Following Eq.~\eqref{eq:FCpatternEn} and using the constraint $\frac34<\lambda_\alpha <\frac54$ for any $\alpha$, we conclude that the energies of all patterns are located in an energy band between $E_0-5A$ and $E_0-3A$.

As in the proof of Property \ref{thm:FCunique4}, we consider separately the case where the configuration has the form $(\bm{v}^\eta,\bm{h})$ for any $\eta$ and $\bm{h}$ is arbitrary, and case where the configuration is of the form $(\bm{v}',\bm{h})$ for any $\bm{v}'$ and $\bm{h}$ is arbitrary.

In the first case we can use the results of the proof of Property~\ref{thm:FCunique1}, namely that (i) pattern $\eta$ has the lowest energy among states $(\bm{v}^\eta,\bm{h})$, and (ii) the energies of hidden-spin excitations above pattern $\eta$ are given by Eq.~\eqref{eq:FCdeltah}.
Taking into account that the energy band of patterns has width $2A$, all the energies of hidden-spin excitations above pattern $\eta$ lay at least $A$ above the energy band of patterns.

In the second case, Property~\ref{thm:FCunique10} states that the lowest energy among configurations $(\bm{v}',\bm{h})$ is equal to $E_0$.
This is $3A$ above the energy band of patterns.
\end{proof}

With these two properties, we can now formulate and proof a theorem of universal representability.
\begin{theorem}
    \label{thm:representing}
    Consider an arbitrary probability distribution, $P_{\textrm{data}}(\bm{v})$, whose domain $\Omega$ is an arbitrary set of configurations $\bm{v}\in\{-1,1\}^V$ of $V$ binary neurons, satisfying $\bm{v}\in\Omega\Leftrightarrow -\bm{v}\in\Omega$.
    A Boltzmann machine with $V$ visible neurons, $H=|\Omega|$ hidden neurons, and weights built from patterns as in Eq.~\eqref{eq:ra} can model arbitrarily close such probability distribution using $K=|\Omega|$ patterns, while having exactly $H$ energy local minima given by the corresponding patterns.
\end{theorem}
\begin{proof}
    The set of all patterns given by Eq.~\eqref{eq:FCpatterns} whose visible parts match elements of $\Omega$ suffices to represent all required visible configurations in our regularised RBM model.
    Therefore one can make probabilities of all other system configurations negligible, which we achieve by taking the limint $A\gg 1$.
    Indeed, Property \ref{thm:FCunique2} states that patterns are energetically separate from all other configurations of the regularised RBM by an energy gap no smaller than $A$.
    Therefore, by making $A$ large enough (which means having a temperature low enough), we can make the probabilities of all neuron configurations other than patterns to be arbitrarily close to zero.
    Also, from Property~\ref{thm:FCunique3} we can conclude that for large enough $A$, any finite ratio of pattern probabilities can be modelled by appropriately choosing the $\lambda_\alpha$ within the allowed range $\frac{3}{4}<\lambda_\alpha<\frac{5}{4}$.
    One can assign the values of the $\lambda_\alpha$ by proceeding in a recursive manner, similar to the approach in Ref.~\cite{le2008representational}.
    Namely, one chooses a value of $\lambda_1$, then adjusts $\lambda_2$ so that the ratio $P_\theta(\bm{v}^1)/P_\theta(\bm{v}^2)$ matches $P(\bm{v}^1)/P(\bm{v}^2)$, etcetera.
    Therefore, for large enough $A$, our regularised RBM construction can be arbitrarily close to any probability distribution for visible configurations. 
\end{proof}

\section{Discussion}\label{sec:discussion}
From the theorems proved in the preceding section follow a number of important features.
We discuss these and their connection to other topics in the field below.

The regularisation of weights through formula Eq.~\eqref{eq:ra} using patterns is the central feature of RAPID BMs.
Superficially, the patterns used resemble the patterns in classical Hopfield networks.
However, they work in a markedly different way.
In our construction, Properties \ref{thm:FCunique1}-\ref{thm:FCunique4} show that patterns given by Eq.~\eqref{eq:FCpatterns} are the \emph{only} local minima of the energy function of the resulting BM, irrespective of the number of patterns (except for the condition for the set $\Omega$) and of the correlations of their visible parts.
Furthermore, from Theorem \ref{thm:perfect} it follows that regularised BMs, at low temperature, have \emph{exponential capacity} for memorisation of visible configurations.
This results seem unusual from the point of using patterns in classical Hopfield networks, but are similar to modern Hopfield networks \cite{Ramsauer2020HopfieldNI}.
Therefore we first discuss the different way patterns work in all these models.

\subsection{Patterns in Boltzmann machines vs. patterns in Hopfield networks}
As mentioned in Section~\ref{sec:intro}, the Hebbian formula \eqref{eq:ra} used for weight regularisation does not guarantee by itself an effective control over the number of local minima of the energy function \eqref{eq:energy}.
In the case of classical Hopfield networks, weights made from {\it random} patterns lead to two distinct phases at low temperatures (i.e., large weights).
These are the ``retrieval'' phase and the spin-glass phase, depending on whether the number of patterns is smaller or greater than a critical value, around $0.14 N$~\cite{Amit1,Amit2}.
In the retrieval phase the patterns that build the weights are ``memorised''.
This means that configurations of the fixed points of the MCMC dynamics of the system are sufficiently close to the patterns, and so a reasonably faithful retrieval of patterns configurations is possible.

Perfect memorisation is realised only for a very small number of patterns.
This can be understood by looking at Eq.~\eqref{eq:ra}: If one has just a single pattern, the weights are correlated such the product of weights for any closed path in a graph describing the network is positive (there is no ``spin frustration'').
This means that one can easily find the ground state of the model by starting from any site and aligning the remaining spins such that every pairwise interaction is energetically minimised.
Such ground state is the chosen pattern or its perfect inversion.

As the number of patterns grows, the minima belonging to ``spurious states'' (composed by pattern hybridisation) appear and proliferate.
Still, the number of minima in the retrieval phase is at most polynomial in $N$.
When the number of patterns exceeds the critical value, the system enters the spin-glass phase, which has a number of energy minima that scales as $e^N$.
Moreover, configurations belonging to these minima are not related to patterns.
This can be understood from the perspective of the central limit theorem: weights composed from random patterns via Eq.~\eqref{eq:ra} become random, independent and Gaussian distributed variables as the number of patterns increases.
In this limit, the correlation between weights decreases drastically and so does the information about particular patterns, while the spin frustration, responsible for occurrence of the spin-glass phase, grows.
This results in the Sherrington-Kirkpatrick spin-glass model~\cite{RevModPhys.58.801}.

At intermediate temperatures the spin-glass phase may extend to a much smaller number of patterns~\cite{Amit1,Amit2} and, at some temperature, even replace the retrieval phase completely.
This is because the proper thermodynamic potential for describing MCMC dynamics is the free energy, which reduces to the energy only at low temperatures.
Therefore, even the control of the number of energy local minima (i.e. having a very small number of patterns) does not guarantee the absence of minima of the free energy function associated to spurious states.
Eventually, at high temperatures, the paramagnetic phase dominated by thermal noise takes over the phase diagram.

It should be noted that in the thermodynamic limit ($N\rightarrow\infty$), the transition from the retrieval phase into the spin-glass phase is not gradual.
The defining feature of the spin-glass phase is the non-zero value of a particular spin-glass order parameter.
Therefore, the difference between the spin-glass and other phases does not simply reduce to the number of low-energy minima.
This means that there is a genuine phase transition, with its singularities in thermodynamic limit, that separates the spin-glass phase from other phases.

The major difference between BMs and classical Hopfield networks is the use of latent neurons in the former, resulting in the universal representing power of the former, scalable with the number of hidden neurons.
Similarly the patterns introduced in Ref.~\cite{rapid} for regularising BMs have hidden parts, in contrast to the patterns in the classical Hopfield networks.
While visible parts of the patterns are related to the data being modelled, the choice for hidden parts is under-determined by data to a large degree.
The only role of the hidden parts of patterns is to regularise weights and control the roughness of the energy landscape.
In Ref.~\cite{rapid}, random hidden parts were used for simplicity, and matched the number of patterns to the complexity of the data to be modeled.
This required to scale the number of hidden neurons such that it was much larger than number of patterns ($H\gg K$) to avoid surpassing the critical value $K_c\sim0.14N$.
In this work we have instead chosen hidden parts that are strongly correlated (recall Eq.~\eqref{eq:FCpatterns}) such that any pair of hidden parts of patterns differs only at two spins.
Furthermore, their form does not depend on the corresponding visible parts.
Therefore the resulting weights do not become random and independent variables even in the case of large number of patterns (see Eq.~\eqref{eq:FCweightsfinal}).
On the contrary, the weights are correlated through the hidden part of the patterns such that even if the number of patterns building them is exponential in $V$, they still are the only local minima of the energy function.
In other words, our current construction does not feature spurious states at low temperatures.
Also, arbitrary correlations of the visible parts of patterns do not lead to additional energy minima forming from hybridisation of patterns, in contrast to classical Hopfield networks.

In Section~\ref{sec:intro} we mentioned that these properties make the proposed model closer to modern Hopfield networks, known also as Dense Associative Memory models \cite{Ramsauer2020HopfieldNI}.
The similarity extends further if one considers the marginalisation of the regularised BM over hidden variables.
This results in multi-spin interactions of very high order between visible spins, a constituting feature of the energy functions of modern Hopfield networks.
After marginalisation, the resulting model retains only the visible parts of patterns.
The hidden parts, importantly, shape the form of the new effective energy function (formally the free energy, if at an arbitrary temperature).
In the proposed construction, the choice of hidden parts lead to an effective energy function
\begin{align}
    \label{eq:dense}
    E(\bm{v})=-\sum_{\eta} \ln\left[2\cosh\left(2A(\bm{v}^\eta\cdot\bm{v}-V+\lambda_\eta)\right)\right],
\end{align}
which closely resembles the function for modern Hopfield networks, $E(\bm{v})=-\sum_{\eta}F(\bm{v}^\eta\cdot\bm{v})$, where $F$ contains high-order powers of its argument \cite{Ramsauer2020HopfieldNI}.
Thus, one can identify the visible parts of the patterns in regularised BMs with patterns in the general formulation of modern Hopfield networks.
Note that the Hebbian form of the weights (recall, Eq.~\eqref{eq:ra}), as well as spin-pair interactions, are not present in this formulation, strengthening the notion that the similarity of our regularisation to patterns in classical Hopfield networks is just superficial.

\subsection{Perfect memorisation of represented configurations}
The aforementioned relation between the visible parts of patterns in regularised BMs and the patterns in modern Hopfield networks suggests that such BMs can perfectly memorise visible configurations as done by the latter.
Indeed, Theorem \ref{thm:perfect} constitutes a positive answer to it.
Namely, it states that regularised BMs can perfectly memorise visible configurations even in the case of strong correlation between them.
Additionally, Theorem \ref{thm:perfect} also implies that regularised BMs have an exponential memorisation capacity.
The cost of this capacity is the one hidden neuron per memorised visible configuration.
In fact, this cost is an upper bound, since our construction shows that memorisation is possible with no assumptions about correlations existent in the data and using shallow RBM models.
Furthermore, the attraction basins for fixed points of zero-temperature MCMC dynamics are very small in our construction (they are of size $V$).

The simple form of the energy function in regularised BMs potentially offers advantage over modern Hopfield networks.
Indeed, although modern Hopfield networks offer exponential capacity for storage of strongly correlated patterns, the computation of their energy function, $E(\bm{v})=-\sum_{\eta}F(\bm{v}^\eta\cdot\bm{v})$, requires a number of operations that scales with the number of used patterns.
When the training data contains millions of configurations, calculating the energy function every time the network configuration $\bm{v}$ changes would make MCMC dynamics very slow.
Even worse, in the case of exponential capacity the energy function would take an exponential number of operations.
The interesting feature of patterns in regularised BMs is that the resultant weights (recall, given by Eq.~\eqref{eq:ra}) do not carry information about individual patterns, but about their sum.
Evaluating the corresponding energy function, Eq.~\eqref{eq:energy} (or Eq.~\eqref{eq:model}), also simple, and involves only interactions between pairs of spins.
One could argue that the large number of operations needed for computing the energy function for modern Hopfield networks is replaced in our regularised BMs by a large number of latent neurons.
However, let us stress that the need of one hidden unit per stored pattern is only an upper bound, obtained for the worst-case scenario of learning a completely random distribution.
In realistic cases data is usually non-random and there are underlying features which can be extracted.
Furthermore, deep networks are particularity effective for such extraction.
Therefore, using our construction for building regularised deep BMs should allow these to cope with realistic data using far less hidden neurons than datapoints to learn.
In such case, operating regularised deep BMs should be much cheaper, in numerical terms, than modern Hopfield networks.
The example of our construction and following Theorem \ref{thm:perfect} lay a basic but strong formal foundation for the memorisation abilities of regularised BMs.

\subsection{Effective regularisation and absence of the spin-glass phase at low temperatures}
Let us now discuss results related to Theorem \ref{thm:representing}, that is, an effective regularisation of BMs while keeping their 
universal representing power. 
In the discussion about patterns in classical Hopfield networks we noted that the Hebbian form of the weights does not guarantee control over the number of minima of the energy function.
Even worse, when one has a very large amount of random patterns the resultant weights are random, independent, and gaussianly distributed.
In this limit, the set of weights can take any form, so using them for regularisation would be meaningless.
Therefore it is crucial that our construction allows for control of the number of minima of the energy function in conjunction with the representing power.
Indeed, Properties \ref{thm:FCunique1}-\ref{thm:FCunique4} show that our regularisation is effective: pattern configurations are local minima of the energy function, there are no energy minima other than the patterns, and we use a reasonable number of hidden neurons to achieve this (equal to the size of dataset $\Omega$).
Here we discuss that our construction, at low temperatures, is also effective in preventing the occurrence of a spin-glass phase.

The low-temperature limit allows us to focus on local minima of the energy function, instead of local minima of the free energy.
Our first argument that regularised BMs cannot be in the spin-glass phase is based on comparing the number of local minima of the corresponding energy function with the number expected in a spin-glass phase.
The latter is expected to grow exponentially with the system size, or in other words, proportionally to the size of the phase space of the system, $2^N$.
The phase space of an RBM is of size $2^{V+H}$, although its bipartite topology leads to an upper bound of $\text{min}(2^V, 2^H)$, as we showed in Lemma~\ref{thm:RBMbound}.
Therefore, for large RBMs in the spin-glass phase the number of minima should be proportional to the smallest quantity of $2^V$ or $2^H$.
Let us therefore consider two cases of dataset size: $|\Omega|=H<V$ and $|\Omega|=H>V$.
In the first case the phase space of RBMs has at most $2^H$ minima, and a finite fraction of that number (i.e. $\propto2^H$) is expected to be the number of minima in the spin-glass phase.
In contrast, in regularised BMs there are only $H$ minima in total.
For a large system (i.e., a system with large $V$ and ratio $H/V$ constant), the discrepancy between $H$ and $\propto2^H$ means that the system is not in a spin-glass phase.
In the second case the phase space of the RBM has at most $2^V$ minima, a finite fraction of which (i.e. $\propto2^V$) is expected to be the number of minima in the spin-glass phase.
If the size of dataset $|\Omega|$ is proportional to some polynomial of $V$, for large systems the discrepancy between $\Omega$ and $\propto2^V$ means that the system is not in a spin-glass phase.
Nevertheless, for very large $\Omega\propto 2^V$ the minima-counting argument fails.

In the case of the maximal $\Omega$ dataset containing all $2^V$ visible configurations, one needs $H=2^V$ hidden neurons to represent it, which means $H=2^V$ patterns and associated energy local minima in our construction --a value which saturates the bound $\text{min}(2^V, 2^H)$.
We therefore now give a different argument, that also applies to the case of $2^V$ minima.
It is based on the fact that in the thermodynamic limit $N\rightarrow\infty$ the defining feature of the spin-glass phase is the non-zero value of a particular spin-glass order parameter, and thus there has to be a sharp transition (non-zero, to zero) from the spin-glass phase to any other phase.
Note that our construction allows us to consider arbitrarily large systems, so the thermodynamic limit is accessible.
For a small to intermediate size of the dataset $|\Omega|$ the system is not in the spin-glass phase.
Pattern retrieval properties and zero-temperature MCMC dynamics (recall the proof of Theorem \ref{thm:perfect}) suggest that the system is, in fact, in the retrieval phase.
As we increase $|\Omega|$ there are no changes in the zero-temperature MCMC dynamics or retrieval properties.
This contradicts a possible phase transition, where MCMC dynamics are expected to change.

In conclusion, our construction is a clear example of effective regularisation of Boltzmann machines resulting in the absence of the spin-glass phase at low temperatures.

\subsection{Universal representing power and sampling hardness}
The novelty of Theorem \ref{thm:representing} lies in proving that universal representability for BMs with weight regularisation is possible while, at the same time, controlling the number of local minima of the energy function.
The ultimate aim of weight regularisation is to ease sampling of BMs, since this is necessary for estimating the gradients during training \cite{Hinton1983proceedings,HintonRBMguidelines,Tieleman2008pcd}.

There are two types of sampling relevant for regularised BMs.
The first type is the standard MCMC sampling, common to any BM, and which can be carried without additional information about the energy landscape of the model.
The second one, specific only to regularised BMs, uses patterns as proxies for low-temperature sampling.
For the standard MCMC sampling, the control over number of minima is useful in reducing the numerical cost since sampling over an exponential number of minima is hard.
Yet, this alone does not guarantee easy sampling.
In the proof of Theorem \ref{thm:representing} we considered the limit $A\gg 1$, in which the probability of any configuration other than patterns was suppressed arbitrarily close to zero.
In such case the energy minima associated with patterns become very deep, while the temperature (given by $A^{-1}$) approaches zero.
In this limit, the system is not easy to sample using standard MCMC dynamics.
This is because the MCMC dynamics is trapped whenever it reaches any pattern configuration, and the transition times between minima corresponding to different patterns grow exponentially.
Fundamentally, the ability of the BM to represent any data probability distribution cannot always be accompanied by requirement of easy sampling using standard MCMC dynamics. The uniform probability distribution for all but one configuration for which there is high probability is the classical example. 
However our construction shows that BM's universal representability can be, instead, accompanied by requirement of controlled number of minima.
We note that the proof for controlled number of minima (Properties \ref{thm:FCunique1}-\ref{thm:FCunique4}) does not require $A\gg 1$.
In this sense Theorem \ref{thm:representing}, which is our main result, proves as much as it can be formally proven regarding universal representing power and sampling hardness of standard MCMC.

However, the practical problem of sampling BMs mostly concerns sampling during training.
First, because training is more difficult and computationally costly than generating, and second, because models successfully trained are usually easy to sample.
We note that RAPID \cite{rapid} contains two key elements: (i) regularisation of weights using patterns (called RA), and (ii) using those patterns as proxies for low-temperature sampling (PID).
The use of patterns for regularising weights not only allows us to avoid the spin-glass phase and hence train in lower temperatures avoiding thermal noise, but the very patterns are additional information we have about the structure of the energy minima in the phase space of the model.
Low-temperature sampling, used for estimation of training parameters, should at least visit the most representative low-energy states.
As shown in Ref.~\cite{rapid}, using the set of patterns either directly as minima proxies or as the seed in MCMC Contrastive Divergence methods allows for efficient pattern-guided sampling of the model averages necessary for training.

The construction presented in this work is the example that regularised BMs have universal representing power in conjunction with easy guided sampling for training at low temperatures.
Indeed, at very low temperatures ($A\gg 1$) the patterns are the only relevant configurations for estimating any averages needed in the training phase.
Hence, using the set of patterns for calculating averages is exact.
Note that having information about patterns eliminates the problems in standard MCMC associated to exponential transition times.
For low temperatures the relevant configurations involve also few-spin excitations from patterns.
Using the set of patterns as a seed for MCMC, Contrastive Divergence methods allow for precise and fast estimation of any relevant averages.

Therefore one can conclude that Theorem \ref{thm:representing} lays the formal foundation for RAPID.
Regularisation of weights (RA) can be efficient (full control over number of minima, using reasonable number of hidden neurons) while allowing for universal representing power.
At the same time, pattern-guided sampling (PID) can be exact at very low temperatures or precise and fast at low temperatures.

\section{Summary and outlook}\label{sec:summary}
Recently, Ref.~\cite{rapid} introduced the concept of weight regularisation in BMs, with the goal of controlling the number of local minima of the energy function and ultimately provide an easier sampling and training for this kind of models.
In this work, we have presented a proof that such models can be universal approximators.
To this end, we have presented a mathematical construction for restricted BMs where their weights are built from patterns, Eq.~\eqref{eq:ra}, ensuring that these patterns are the only local minima of the energy function.
This is irrespective of the correlations between the visible parts of patterns as well as of its number, which may be even exponential in the number of visible neurons of the model.
Our construction allows to assign arbitrary probabilities to such patterns, while reducing the probabilities of all other states arbitrarily close to zero.
The universal representability of the construction is the proved by the fact that the visible parts of the patterns can cover any set of visible configurations.
Moreover, the model exhibits fast and precise pattern-guided sampling at low temperatures.
These two results, universal representative power and fast sampling, lay the formal foundations of the RAPID method presented in Ref.~\cite{rapid}.

Furthermore, the construction proposed allows also to go beyond the practical application of RAPID showcased in Ref.~\cite{rapid}.
We have shown that correlations between the hidden parts of the patterns may radically decrease the number of hidden neurons needed to avoid pattern hybridisation.
Indeed, with a number of hidden neurons just equaling the number of patterns in the model, our construction avoids hybridisation and the spin-glass phase even in case of an exponential number of patterns.
The exploration of the practical use of correlated hidden components within patterns in RAPID models in real-world experiments is a prospect reserved for future research.

Additionally, our construction allows to connect regularised BMs with modern Hopfield networks.
In Theorem~\ref{thm:perfect} we proved that RA-BM can be used for memorisation of the visible configurations with exponential capacity.
While marginalisation over the hidden neurons makes RA-BMs similar to modern Hopfield networks, the form of unmarginalised RA-BMs presents a much simpler and less numerically costly energy function.
Therefore we believe that the deep regularised BMs may be an interesting complement to modern Hopfield networks.
Experiments with such machines are underway.

The main drawback of our construction as an example of memorising RA-BM is the very small attraction basin for fixed points in zero-temperature MCMC dynamics.
However we have found yet another mathematical construction of patterns that has the universal representative power in one limit and arbitrarily large attraction basins in another.
Furthermore such construction better reflects spectral form of BMs trained with real data, however it is more difficult for mathematical analysis.
We leave description and analysis of this construction for another work.

\begin{acknowledgments}
ICFO group acknowledges support from: ERC AdG NOQIA; Ministerio de Ciencia y Innovation Agencia Estatal de Investigaciones (PGC2018-097027-B-I00/10.13039/501100011033, CEX2019-000910-S/10.13039/501100011033, Plan National FIDEUA PID2019-106901GB-I00, FPI, QUANTERA MAQS PCI2019-111828-2, QUANTERA DYNAMITE PCI2022-132919, Proyectos de I+D+I ``Retos Colaboración'' QUSPIN RTC2019-007196-7); MICIIN with funding from European Union NextGenerationEU (PRTR-C17.I1) and by Generalitat de Catalunya; Fundació Cellex; Fundació Mir-Puig; Generalitat de Catalunya (European Social Fund FEDER and CERCA program, AGAUR Grant No. 2021 SGR 01452, QuantumCAT \ U16-011424, co-funded by ERDF Operational Program of Catalonia 2014-2020); Barcelona Supercomputing Center MareNostrum (FI-2023-1-0013); EU (PASQuanS2.1, 101113690); EU Horizon 2020 FET-OPEN OPTOlogic (Grant No 899794); EU Horizon Europe Program (Grant Agreement 101080086 — NeQST), National Science Centre, Poland (Symfonia Grant No. 2016/20/W/ST4/00314); ICFO Internal ``QuantumGaudi'' project; European Union’s Horizon 2020 research and innovation program under the Marie-Skłodowska-Curie grant agreement No 101029393 (STREDCH) and No 847648  (``La Caixa'' Junior Leaders fellowships ID100010434: LCF/BQ/PI19/11690013, LCF/BQ/PI20/11760031,  LCF/BQ/PR20/11770012, LCF/BQ/PR21/11840013).
E.P. is supported by ``Ayuda (PRE2021-098926) financiada por MCIN/AEI/ 10.13039/501100011033 y por el FSE+''.
A.P.-K. acknowledges support from the Spanish Ministry of Science and Innovation MCIN/AEI/10.13039/501100011033 (CEX2019-000904-S and PID2020-113523GB-I00), the Spanish Ministry of Economic Affairs and Digital Transformation (project QUANTUM ENIA, as part of the Recovery, Transformation and Resilience Plan, funded by EU program NextGenerationEU), Comunidad de Madrid (QUITEMAD-CM P2018/TCS-4342), Universidad Complutense de Madrid (FEI-EU-22-06), the CSIC Quantum Technologies Platform PTI-001, and the NCCR SwissMAP.
G.M-G. acknowledges funding from the European Union.
Views and opinions expressed are, however, those of the author(s) only and do not necessarily reflect those of the European Union, European Commission, European Climate, Infrastructure and Environment Executive Agency (CINEA), nor any other granting authority.
Neither the European Union nor any granting authority can be held responsible for them. M.Á.G.-M. acknowledges funding from the Spanish Ministry of Education and Professional Training (MEFP)
through the Beatriz Galindo program 2018 (BEAGAL18/00203), QuantERA II Cofund
2021 PCI2022-133004, Projects of MCIN with funding from European Union NextGenerationEU (PRTR-C17.I1) and by Generalitat Valenciana, with Ref. 20220883 (PerovsQuTe)
and COMCUANTICA/007 (QuanTwin), and Red Temática RED2022-134391-T.

\end{acknowledgments}

\bibliographystyle{apsrev4-1}
\bibliography{bibliography}

\end{document}